\documentclass[twocolumn,twosde]{IEEEtran} 
\usepackage{color}
\usepackage{amsmath, amssymb, amsthm, amsfonts}
\usepackage[final]{graphicx}
\usepackage{dsfont}
\usepackage{url}
\usepackage[latin1]{inputenc}
\usepackage{multirow}
\usepackage{array}
\usepackage{color}
\usepackage[colorlinks,citecolor=blue,urlcolor=blue]{hyperref}
\usepackage{stfloats}
\usepackage[numbers,sort&compress]{natbib}
\usepackage{hhline}

\newcommand{\e}{\mathrm{e}}

\DeclareMathOperator\Ei{Ei}

\newtheoremstyle{definition}%
  {0.5em}
  {\topsep}
  {\upshape}
  {}
  {\bfseries}
  {.}
  { }
  {}

\theoremstyle{plain}
\newtheorem{theorem}{Theorem}

\newtheorem{corollary}{Corollary}
\newtheorem{remark}{Remark}
\theoremstyle{definition}

\def \a {{a_\kappa}}
\def \b {{b_\kappa}}

\makeatletter

\def\blfootnote{\xdef\@thefnmark{}\@footnotetext}
\makeatother

\begin{document}

\title{Exact Statistical Characterization of $2\times2$ Gram Matrices with Arbitrary Variance Profile}

\author{N. Auguin, D. Morales-Jimenez, M. R. McKay

\thanks{Copyright~\copyright~2015 IEEE. Personal use of this material is permitted. However, permission to use this material for any other purposes must be obtained from the IEEE by sending a request to pubs-permissions@ieee.org.}
\thanks{The authors are with the Department of Electrical and Computer Engineering, Hong Kong University of Science and Technology, Hong Kong. E-mail: nicolas.auguin@connect.ust.hk.}
\thanks{This work was supported by the Hong Kong RGC General Research Fund  under grant number 616713, and the Hong Kong Telecom Institute of Information Technology under grant number HKTIIT16EG01.}
}


\maketitle

\begin{abstract}
This paper is concerned with the statistical properties of the Gram matrix $\mathbf{W}=\mathbf{H}\mathbf{H}^\dagger$, where $\mathbf{H}$ is a $2\times2$ complex central Gaussian matrix whose elements have arbitrary variances. With such arbitrary variance profile, this random matrix model fundamentally departs from classical Wishart models and presents a significant challenge as the classical analytical toolbox no longer directly applies. We derive new exact expressions for the distribution of $\mathbf{W}$ and that of its eigenvalues by means of an explicit parameterization of the group of unitary matrices.
Our results yield remarkably simple expressions, which are further leveraged to study the outage data rate of a dual-antenna communication system under different variance profiles.
\end{abstract}

\begin{IEEEkeywords}
Random matrix theory, MIMO channels, eigenvalue distribution.
\end{IEEEkeywords}

\IEEEpeerreviewmaketitle

\section{Introduction}
This paper investigates the statistical properties of random matrices of the form $\mathbf{W} = \mathbf{H} \mathbf{H}^\dagger$, where $\mathbf{H}$ is $2 \times 2$ with independent entries  
\begin{align} \label{eq:Basic}
\mathbf{H}_{ij}  \sim \mathcal{C N} (0,  \phi_{ij}) , \quad \quad i, j = 1, 2 \; .
\end{align}
The distinguishing feature is that the variance profile, $\{\phi_{ij}\}_{i,j=1,2}$, is allowed to be \emph{arbitrary}.  

Despite its apparent simplicity, it is remarkable that little is known about the statistical properties of such matrices, beyond specific examples. Most notable is the case where the variances factorize as $\phi_{i j} = \sigma_i \pi_j$, where the model bears a strong analogy with so-called ``Kronecker correlated'' models that have been studied extensively in communication theory (see, for example, \cite{hanlen2003capacity,shin2006capacity}) as well as in classical statistics (see, for example, \cite{James1964,muirhead2009aspects}). Such Kronecker models, as well as their numerous adaptations or extensions (e.g., \cite{jayaweera2003performance,wang2005capacity,jin2008transmit}), enjoy certain symmetry properties that allow their characterization by leveraging classical tools in multi-variate analysis, such as known matrix-variate integrals, hypergeometric functions of matrix arguments and zonal polynomials \cite{James1964,muirhead2009aspects,mehta2004random}. The model in (\ref{eq:Basic}) is fundamentally different, in that it does not readily lend itself to analysis via these classical techniques.

From a communication engineering perspective, models of the form (\ref{eq:Basic}) are useful since they can suitably characterize channels between multiple transmit and receive antennas that are arbitrarily distributed in space. These may include, among others, the so-called distributed antenna systems (DAS), which have recently attracted interest within the wireless communications community \cite{zhang2004capacity,saleh1987distributed,roh2002outage}. Despite the interest of DAS and the trends towards ever more heterogeneous and distributed network architectures, a precise understanding of such systems remains outstanding, due in part to the scarcity of statistical results on the underlying random matrix model. Such results have mainly been established in the asymptotic regime where the dimensions of the random matrix $\mathbf{H}$ grow large (see, \cite{zhang2004capacity,hachem2008clt}). These asymptotic results are rather complex and serve primarily as approximations for large-dimensional systems whose behavior may differ from that of finite ones.

In this paper, we present an exact characterization of random matrix models with arbitrary variance profiles, deriving for the first time new exact expressions for the joint distribution of (i) the random matrix $\mathbf{W}$, and (ii) its eigenvalues. While we focus on the $2 \times 2$ case, we demonstrate that the analysis is still rather complicated. A main challenge encountered in the derivations is that they involve the computation of certain integrals with respect to the group of $2 \times 2$ unitary matrices. These integrals are not classical, and we solve them by working with an explicit parametrization of the unitary group. Despite the complexity of the derivations, our results yield an exact and remarkably simple expression for the matrix density, along with a tractable expression for the eigenvalue density which reduces to particularly simple forms for various choices of the variance profile. Building upon these results, we further derive simple expressions for the distribution of the extreme eigenvalues, which are then leveraged to study the outage data rate of a dual-user multi-antenna communication system under different variance profiles. In particular, we show that asymmetry in the variance profile can significantly degrade the outage rate of systems with distributed antennas.

\section{Main results} \label{sec:main} 
This section presents our key mathematical results. 

\label{sec:gen}
\begin{theorem} \label{th:W}
Consider $\mathbf{W}=\mathbf{H}\mathbf{H}^{\dagger}=\bigl( \begin{smallmatrix} 
  w_{1} & w_{3}\\
  w_{3}^\star & w_{2} 
\end{smallmatrix} \bigr) \succeq \mathbf{0}$, with $\mathbf{H}$ as in (\ref{eq:Basic}), with $\phi_{ij}>0$ for $i,j=1,2$. Assume $\phi_{i1} \neq \phi_{i2}$ for some $i$. The probability density function (PDF) of $\mathbf{W}$ admits
\begin{align}
p(\mathbf{W}) &= \frac{K}{\pi} \, e^{-\frac{1}{2}\left(w_{1}s_{1}+w_{2}s_{2}\right)}\nonumber  \\
& \times
\frac{\sinh{\left(\frac{1}{2}\sqrt{(w_{1}\epsilon_{1}-w_{2}\epsilon_{2})^2+
			4|w_{3}|^{2}\epsilon_{1}\epsilon_{2}}\right)}}{\frac{1}{2}\sqrt{(w_{1}\epsilon_{1}-w_{2}\epsilon_{2})^2+4 |w_{3}|^{2}\epsilon_{1}\epsilon_{2}}} \label{eq:pw1}
\end{align}
where $K=\prod_{1\le i,j \le 2}\frac{1}{\phi_{ij}}$, $
s_{i}=\frac{1}{\phi_{i1}}+\frac{1}{\phi_{i2}}$, and
$\epsilon_{i}=\frac{1}{\phi_{i1}}-\frac{1}{\phi_{i2}}$.
\end{theorem}

\begin{proof}
See Appendix \ref{Ap:Th1Proof} for a complete proof. Briefly: $\mathbf{H}$ is decomposed as $\mathbf{H} = \mathbf{LQ}$, with $\mathbf{L}$ lower triangular and $\mathbf{Q}$ unitary, and after applying the corresponding Jacobian and integrating over the unitary group to eliminate $\mathbf{Q}$, we obtain the PDF of $\mathbf{L}$. Applying the variable transformation $\mathbf{W}=\mathbf{H}\mathbf{H}^{\dagger}=\mathbf{L}\mathbf{L}^{\dagger}$ 
leads to the result. 
\end{proof}

\begin{theorem} \label{th:lambda} Assume $\epsilon_i \neq 0$ for some $i$. The joint PDF of the (ordered) eigenvalues $\lambda_1 \ge \lambda_2 >0$ of $\mathbf{W}$ admits
\begin{align}
p(\lambda_{1},\lambda_{2})&= 2 K \, (\lambda_{1}-\lambda_{2})^{2}  \int_{0}^{\frac{\pi}{2}}e^{-\frac12 \nu(\lambda_{1},\lambda_{2},\kappa)} \, \nonumber \\
& \times \frac{\sinh\left(\frac12 \sqrt{\eta(\lambda_{1},\lambda_{2},\kappa)}\right)}{\sqrt{\eta(\lambda_{1},\lambda_{2},\kappa)}}\sin(2\kappa)d\kappa  \, ,
\label{eq:plambda_integral}
\end{align}
where
\begin{align*}
	\nu(\lambda_{1},\lambda_{2},\kappa) &= (\a \lambda_{1}+ \b \lambda_{2}) s_{1}
	+(\b \lambda_{1} + \a \lambda_{2} ) s_{2}  , \\
	\eta(\lambda_{1},\lambda_{2},\kappa) &=
	(\a \lambda_{1} + \b \lambda_{2})^{2} \epsilon_{1}^{2}
	+ (\b \lambda_{1} + \a \lambda_{2})^{2} \epsilon_{2}^{2} \nonumber \\
	&  +2(\a \b (\lambda_{1}-\lambda_{2})^{2}-\lambda_{1}\lambda_{2}) \epsilon_{1}\epsilon_{2} \,\, ,
\end{align*}
with $\a =\cos^{2}(\kappa)$ and $\b=\sin^{2}(\kappa)$ .
\end{theorem}

\begin{proof}
See Appendix \ref{Ap:Th2Proof}.
\end{proof}

\begin{remark}[Equivalence of the variance profile]
\label{rem:sym}
Let $\mathbf{\Phi}=(\phi_{ij})_{i,j=1,2}$ be the matrix defining the variance profile associated with $\mathbf{W}$. Since $\mathbf{H}\mathbf{H}^\dagger$ and $\mathbf{H}^\dagger\mathbf{H}$ share the same eigenvalues, it is equivalent to consider $\mathbf{\Phi}$ or $\mathbf{\Phi}^\mathrm{T}$. 
\end{remark}

\subsection{Partially asymmetric variances}
Our main results---the PDF of $\mathbf{W}$ and of its eigenvalues--- have been given for a general variance profile. Consider now the special case where asymmetry in the variances is only partially allowed; specifically, consider $\phi_{i1} = \phi_{i2}$ for some $i$. 
Assume without loss of generality (by symmetry of the PDF) that $\phi_{21} = \phi_{22}$ and, therefore, $\epsilon_2 = 0$.

\begin{corollary} \label{cor:epsilon}
	Consider the case $\phi_{21} = \phi_{22} \triangleq \phi_3$ and define $\phi_1 \triangleq \min(\phi_{11},\phi_{12})$, $\phi_2 \triangleq \max(\phi_{11},\phi_{12})$, with $\phi_1 \neq \phi_2$ (hence $\epsilon_1 \neq 0$). The PDF of $\mathbf{W}=\mathbf{H}\mathbf{H}^{\dagger} \succeq \mathbf{0}$ admits
	\begin{align}
	\label{eq:w_part}
	p(\mathbf{W})= \frac{K}{\pi} e^{-\frac{1}{2}\left(w_{1}s_{1}+w_{2}s_{2}\right)} \,
	\frac{\sinh{\left(\frac{1}{2}w_1|\epsilon_1|\right)}}{\frac{1}{2}w_1|\epsilon_1|} ,
	\end{align}
	and the joint PDF of its (ordered) eigenvalues reduces to
	\begin{align}
	p(\lambda_{1},\lambda_{2})=& \nonumber \\
	  & \hspace{-1.3cm} \frac{ \, e^{-\frac{\lambda_1+\lambda_2}{\phi_3}}  }{(\phi_2-\phi_1)\phi_3^2} \det  \left(\lambda_i^{j-1} \right)_{i,j=1,2} \det \left( g(\lambda_j)^{i-1} \right)_{i,j=1,2}, 
	 \label{eq:lambda_part_case}
\end{align}
with $g(x) \triangleq \Ei((1/\phi_3-1/\phi_2)x) - \Ei((1/\phi_3-1/\phi_1)x)$ and $\Ei(x) = - \int_{-x}^{\infty} \frac{e^{-t}}{t}dt$ the exponential integral function.

Furthermore, the cumulative distribution function (CDF) of the minimum eigenvalue of $\mathbf{W}$ and the CDF of the maximum eigenvalue of $\mathbf{W}$ are given in (\ref{eq:Fmin}) and (\ref{eq:Fmax}) (top of the next page) for $\phi_1, \phi_2, \phi_3$ all distinct.

\end{corollary}

\begin{figure*}[!]
\normalsize
\setcounter{equation}{5}
\begin{equation}
\label{eq:Fmin}
F_{\lambda_\mathrm{min}}(x) = \mathbb{P}(\lambda_\mathrm{min}\le x) =1-\frac{e^{-x/\phi_3}}{\phi_2 - \phi_1}\left( \phi_2e^{-x/\phi_2} - \phi_1e^{-x/\phi_1} + x \left( \Ei(-x/\phi_2)- \Ei(-x/\phi_1) \right)\right)
\end{equation}
\begin{align}
\label{eq:Fmax}
F_{\lambda_\mathrm{max}}(x)=\mathbb{P}(\lambda_\mathrm{max}\le x)
 &= \frac{1}{\phi_2-\phi_1} \biggl((1-e^{-x/\phi_3})\left(\phi_2(1-e^{-x/\phi_2})-\phi_1(1-e^{-x/\phi_1}) \right) \biggr. \nonumber\\
& \hspace{+1mm} \biggl. + x e^{-x/\phi_3}\left(-g(x)+\Ei(-x/\phi_2)-\Ei(-x/\phi_1)+\log \left| \frac{\phi_3-\phi_2}{\phi_3-\phi_1}\right|\right) \biggr)
\end{align}	
\setcounter{equation}{7}
\vspace*{4pt}
\hrulefill
\end{figure*}

\begin{proof}
A sketch of the proof is given in Appendix \ref{app:sec}.
\end{proof}

Note the remarkable simplicity of both the joint eigenvalue PDF and the marginal CDFs of the extreme eigenvalues in this special case, which retain in part the flexibility of the general model, with $3$ arbitrary variances rather than $4$.

\begin{remark}[On the tail of the extreme eigenvalue distribution]
\label{rem:F_exp}
In the setting of \emph{Corollary} \ref{cor:epsilon}, we have the following expansions for $x$ in the neighborhood of $0$:
\begin{align}
\label{eq:fmin_exp}
F_{\lambda_\mathrm{min}}(x) &=  \left(\frac{1}{\phi_3}+\frac{\log \phi_2-\log \phi_1}{\phi_2-\phi_1} \right) x + o(x) , \\
\label{eq:fmax_exp}
F_{\lambda_\mathrm{max}}(x) &= \frac{1}{12}\frac{1}{\phi_1 \phi_2 \phi_3^2}x^4  + o(x^4).
\end{align}
These results are obtained by basic algebra, and making use of \cite[eq. 8.214]{gradshteyn1965table}. The expressions are remarkably simple and shed light on how the variance profile affects the tail of the extreme eigenvalue distributions. For example, assuming the total variance is normalized, (\ref{eq:fmax_exp}) suggests that a strong asymmetry in the variance profile---i.e., some variances substantially smaller than others---leads to a more heavy-tailed distribution $F_{\lambda_\mathrm{max}}$, as compared to a more uniform profile. The insights brought by these simple expressions are further illustrated in Section \ref{sec:app}, where we use $F_{\lambda_\mathrm{min}}$ to study the outage data rate of a communication system with distributed antennas.    
\end{remark}

\subsection{Connection to Kronecker correlated models}
It is instructive to relate the random matrix $\mathbf{W}$ to Kronecker correlated models, which are commonly considered in multi-antenna communications \cite{hanlen2003capacity,mckay2007performance,shin2006capacity,chiani2003capacity}.  For such models, the channel matrix can be described as
\begin{align*}
\mathbf{H}_K = \mathbf{R}^{1/2}\mathbf{H}_w \mathbf{T}^{1/2},
\end{align*}
where $\left(\mathbf{H}_w \right)_{ij}$ are independent $\mathcal{CN}(0,1)$, while $\mathbf{R}$ and $\mathbf{T}$ are non-negative definite. Denoting $\mathbf{U}_\mathbf{R}$ (resp.\@ $\mathbf{U}_\mathbf{T}$) an eigenbasis of $\mathbf{R}$ (resp.\@ $\mathbf{T}$) and $r_i$ (resp.\@ $t_i$) the $i$-th eigenvalue of $\mathbf{R}$ (resp.\@ $\mathbf{T}$), $\mathbf{H}_K$ can be written (in the $2\times 2$ case) as \cite{weichselberger2006stochastic}
\begin{align*}
\mathbf{H}_K = \mathbf{U}_\mathbf{R} \left( {r_1^{1/2} \choose r_2^{1/2}}(t_1^{1/2}, t_2^{1/2}) \odot \mathbf{H}_w \right) \mathbf{U}_\mathbf{T},
\end{align*}
where $\odot$ denotes the Hadamard (entry-wise) matrix product.
  
Furthermore, it can be shown that the eigenvalue distribution of $\mathbf{W}_K$, described as $\mathbf{W}_K =  \mathbf{H}_K \mathbf{H}_K^\dagger$, depends on $\mathbf{R}$ and $\mathbf{T}$ only through their eigenvalues (see \cite{hanlen2003capacity,mckay2007performance,shin2006capacity}). Hence our result in \emph{Theorem} \ref{th:lambda} subsumes the eigenvalue distribution of Kronecker correlated models as a special case. 

\section{Outage performance of a dual-user communication system with distributed antennas}\label{sec:app}
We now demonstrate the usefulness of the mathematical results exposed above, through a concrete communications application example. Consider a communication system in which 2 single-antenna users (transmitters) communicate with a receiver comprising 2 distributed antennas. Rayleigh fading is assumed, with shadowing neglected, in which the communication channel is of the form $\mathbf{H}$ in (\ref{eq:Basic}) with variance profile $\phi_{ij}=D_{ij}^{-\nu}$, where $\nu$ is the path loss exponent and  $D_{ij}$ the distance between transmit antenna $j$ and receive antenna $i$. Thus, the placement of the antennas determines the channel variance profile. For instance, if both transmitters (i.e., users) are located at equal distance from receive antenna $i$, then $\phi_{i1}=\phi_{i2}$, which corresponds to the setting of \emph{Corollary} \ref{cor:epsilon}. 

We further assume that the receiver has perfect knowledge of $\mathbf{H}$, while the transmitters ignore such knowledge and send independent data with a total transmit power $P$.  The noise at each receiver is assumed independent $\mathcal{C N}(0, \sigma_n^2$), and we define the transmit signal to noise ratio (SNR) as $\rho \triangleq P /\sigma_n^2$. We further assume that the total power gain of the channel is fixed, with $\mathbb{E} \left[ {\rm tr}\left( \mathbf{H} \mathbf{H}^\dagger \right) \right]= \sum_{1 \le i,j \le 2} \phi_{ij} =4$.

Denoting $\mathbf{x}$ the vector of transmitted signals and $\mathbf{n}$ the additive noise, the received signal $\mathbf{y}$ takes the form
\begin{align*}
\mathbf{y}=\mathbf{H}\mathbf{x} + \mathbf{n}.
\end{align*} 
For detection, a linear zero-forcing receiver is considered. Such receivers are popular because of their low complexity \cite{gore2002transmit}, and their performance is known to approach that of minimum-mean squared error receivers at high SNR \cite{kumar2009asymptotic}. The estimate $\hat{\mathbf{x}}$ of the transmitted signal $\mathbf{x}$ then becomes
\begin{align*}
\hat{\mathbf{x}} =  \left( \mathbf{H}^\dagger\mathbf{H} \right)^{-1} \mathbf{H}^\dagger \mathbf{ y} = \mathbf{x} + \left(\mathbf{H}^\dagger\mathbf{H}  \right)^{-1} \mathbf{H}^\dagger  \mathbf{n},
\end{align*}
and the post-processing SNR for the $i$-th user is \cite{heath2005multimode}
\begin{align*}
\mathrm{SNR}_i =  \frac{\rho}{\left[\left(\mathbf{H}^\dagger \mathbf{H}\right)^{-1}\right]_{ii}} 
= \frac{\rho}{\left[\mathbf{W}^{-1}\right]_{ii}}.
\end{align*}

Here, we are interested in the outage data rate, defined as the largest transmission rate (in bits/s/Hz) that can be reliably guaranteed for both users (simultaneously) at least $(1-\epsilon)\times 100\%$ of the time, i.e.,
\begin{align}
R_\mathrm{out}(\epsilon)=\underset{R\ge0}{ \sup} \left(R:P_\mathrm{out}(R) < \epsilon\right),
\label{eq:def_r}
\end{align}
with $\epsilon$ being the prescribed maximum outage level, and $P_\mathrm{out}(R)$ denoting the outage probability for a given target rate $R$. That is, $P_\mathrm{out}(R)$ reflects the probability that a reliable transmission at rate $R$ cannot be guaranteed to both users, given by
\begin{align*}
P_\mathrm{out}(R) = \mathbb{P}\left(\log_2\left(1+\mathrm{SNR}_\mathrm{min} \right) \le R \right) ,
\end{align*}
where $\mathrm{SNR}_\mathrm{min}=\min(\mathrm{SNR}_1,\mathrm{SNR}_2)$.  Since \cite{heath2005multimode} $\frac{1}{\left[\mathbf{W}^{-1}\right]_{ii}} \ge  \lambda_\mathrm{min}, i = 1, 2,$  
with $\lambda_\mathrm{min}$ the minimum eigenvalue of $\mathbf{W}$, it follows that $\mathrm{SNR}_\mathrm{min} \ge \rho \lambda_\mathrm{min}$, which yields the upper bound
\begin{align*}
P_\mathrm{out}(R) &\le  \mathbb{P} \left(\log_2\left(1+\rho\lambda_\mathrm{min} \right) \le R \right) \nonumber \\
&= F_{\lambda_\mathrm{min}}\left(\frac{1}{\rho}  \left(2^R-1\right) \right) \; .
\end{align*}
Considering now the setting of \emph{Corollary} \ref{cor:epsilon} ($\phi_{21}=\phi_{22}$) and a small outage level, we can use the expansion of $F_{\lambda_\mathrm{min}}(\cdot)$ (Remark \ref{rem:F_exp}) to write
\begin{align*}
P_\mathrm{out}(R) 
\lesssim \frac{1}{\rho}\left(2^{R}-1\right) \underbrace{\left(\frac{1}{\phi_3}+\frac{\log \phi_2-\log \phi_1}{\phi_2-\phi_1} \right)}_{a(\mathbf{\Phi})} \; ,
\end{align*}
where the influence of the variances is isolated through the factor $a(\mathbf{\Phi})$.  Also note that if $\phi_2 \to \phi_1$, $a(\mathbf{\Phi}) \to 1/\phi_3 + 1/\phi_1$.

The results above then suggest the following lower bound for the outage data rate $R_\mathrm{out}(\epsilon)$:
\begin{align}
R_\mathrm{out}(\epsilon) \ge \log_2 \left( 1+ \rho F_{\lambda_\mathrm{min}}^{-1}(\epsilon) \right) \triangleq \check {R}_\mathrm{out}(\epsilon),
\label{eq:cout_exp}
\end{align}
where $F_{\lambda_\mathrm{min}}^{-1}$ denotes the inverse function of $F_{\lambda_\mathrm{min}}$, while in the setting of \emph{Corollary} \ref{cor:epsilon}, 
\begin{align}
R_\mathrm{out}(\epsilon) \gtrsim  \log_2 \left(1+ \rho \frac{\epsilon}{a(\mathbf{\Phi})} \right) \triangleq {\tilde R}_\mathrm{out}(\epsilon)
. 
\label{eq:cout_exp2}
\end{align}

Fig.\@ \ref{fig:r_vs_snr} shows the outage data rate as a function of the SNR $\rho$, for an asymmetric variance profile $\mathbf{\Phi}=\left( \begin{smallmatrix} 
  0.01 & 0.99\\
  1.5 & 1.5 
\end{smallmatrix} \right)$, and three maximum outage levels $\epsilon=1\%, 10\%$ and $50\%$. For each maximum outage level, we plot (i) the true empirical rate $R_\mathrm{out}(\epsilon)$ in (\ref{eq:def_r})---obtained from the average over $10^5$ realizations of the channel, (ii) the lower bound $\check {R}_\mathrm{out}(\epsilon)$ in (\ref{eq:cout_exp})---where we invert $F_{\lambda_\mathrm{min}}$ from (\ref{eq:Fmin}) numerically, and (iii) the approximate analytical lower bound ${\tilde R}_\mathrm{out}(\epsilon)$ in (\ref{eq:cout_exp2}).

\begin{figure}[h]
        \centering
        \includegraphics[width=0.85\columnwidth]{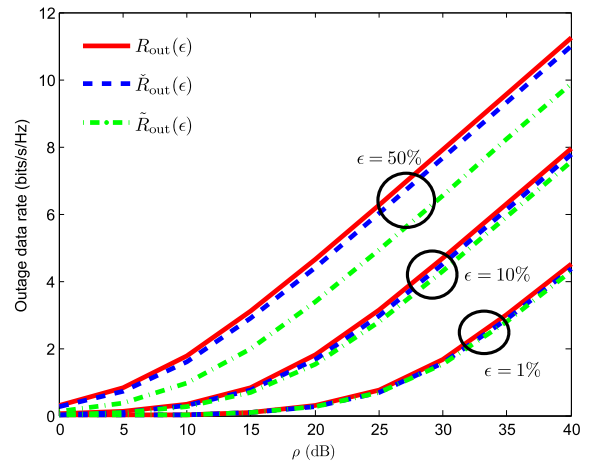}
        \caption{\small Outage data rate vs.\@ SNR for different maximum outage levels, $\epsilon=1\%, 10\%$ and $50\%$, with an asymmetric variance profile.}
\label{fig:r_vs_snr}
\end{figure}

Notice the tightness of the lower bound $\check R_\mathrm{out}(\epsilon)$ for all the considered maximum outage levels and all SNRs. Moreover, the analytical (approximate) bound $\tilde R_\mathrm{out}(\epsilon)$, based on the expansion (\ref{eq:fmin_exp}), shows excellent accuracy for small outage levels ($\epsilon=1\%, 10\%$) as expected, since it corresponds to evaluating $F_{\lambda_\mathrm{min}}(x)$ fairly deep in the tail. However, for higher outage levels (e.g., $\epsilon=50\%$), this bound becomes less reliable.
 
For $\epsilon$ small, the effect of the variance profile on the outage data rate can be analyzed from the approximate bound $\tilde R_\mathrm{out}(\epsilon)$, given its remarkable tightness and simplicity. 
Recalling the normalization on the total power gain of the channel, i.e., $\phi_1+\phi_2 = 4-2\phi_3$, it is straightforward to verify that $a(\mathbf{\Phi})$ is continuous in the parameters $\phi_1,\phi_2,\phi_3$,
and that, for $\phi_3$ fixed, the mapping $(\phi_1,\phi_2) \mapsto a(\mathbf{\Phi})$ is minimum when $\phi_1=\phi_2$ and maximum in the limit $\phi_1 \rightarrow 0$ (hence $\phi_2 \rightarrow 4-2\phi_3$). This immediately implies that, for $\phi_3$ fixed, the outage data rate is maximum when $\phi_1=\phi_2=2-\phi_3$, which represents  the ``most symmetric'' profile under the total variance normalization, and that any departure from such symmetry entails a performance loss. To quantify the range of such loss, we now consider the two extremes cases, ``symmetric'' $\mathbf{\Phi}^\mathrm{sym} = ( \begin{smallmatrix} 
2-\phi_3 & 2-\phi_3\\
\phi_3 & \phi_3 
\end{smallmatrix} )$ and ``asymmetric'' $\mathbf{\Phi}^\mathrm{asym} = ( \begin{smallmatrix} 
0.01 & 4-2\phi_3-\phi_1\\
\phi_3 & \phi_3 
\end{smallmatrix} )$ profiles, and define the fractional loss in the outage data rate due to asymmetry, i.e., ${ \tilde{FL}}(\epsilon) \triangleq \frac{{\tilde R}_\mathrm{out}^\mathrm{sym}(\epsilon)-{\tilde R}_\mathrm{out}^\mathrm{asym}(\epsilon)}{{\tilde R}_\mathrm{out}^\mathrm{sym}(\epsilon)}$.
In Table \ref{tab:r_vs_phi}, for $\epsilon = 1\%$, we report the value of $\tilde{FL}(\epsilon)$ for increasing values of $\phi_3 \in (0,2)$, along with the corresponding true fractional loss $FL (\epsilon)$, computed using the true outage data rates, obtained empirically by averaging over $10^5$ realizations. The numbers reveal a striking degradation in the outage data rate due to asymmetry, e.g., up to $36\%$ loss for $\phi_3 = 1.6$ and SNR$=30$dB.

\renewcommand{\arraystretch}{1.5}
\begin{table}[htb]
\caption{Fractional Loss of the Outage Data Rate Associated with Asymmetric Variance Profiles, at SNR$=30\mathrm{dB}$ and $\epsilon=1\%$.}
\label{tab:r_vs_phi}
\begin{center}
\begin{tabular}{c|cccccccc}
   $\phi_3$ &  0.01 & 0.5 & 1 & 1.2 & 1.4 & 1.6 & 1.8 & 1.95 \\
\hline
$FL  (\epsilon)$ & 0\% & 18\% & 28\% & 30\% & 32\% &\bf{36\%} & 33\% & 24\%\\
${ \tilde{FL}}(\epsilon)$  & 1\% & 19\% & 27\% & 29\% & 33\% & \bf{33\%} & 33\% & 23\%\\
\end{tabular}
\end{center}
\end{table}

Physically, it implies that, assuming the two transmitting users are equidistant from receive antenna 2, their position relative to receive antenna 1 is crucial: if the distances from the two users to receive antenna 1 are very different (asymmetric case, $\phi_2/\phi_1 \gg 1$), a significantly lower outage data rate is expected as compared to the case where both users are equidistant from receive antenna 1 (symmetric case, $\phi_1=\phi_2$).

\appendix

\subsection{Proof of Theorem \ref{th:W}}  \label{Ap:Th1Proof}

We start with the \textit{LQ} decomposition $\mathbf{H}=\mathbf{LQ}$, with $\mathbf{L}$ lower triangular and $\mathbf{Q}$ unitary. $\mathbf{L}$ and $\mathbf{Q}$ are parametrized as \cite{zyczkowski1994random,fraidenraich2007mimo,fraidenraich2008mimo}: $$\mathbf{L} = \left( \begin{smallmatrix} 
  l_{11} & 0\\
  l_{21R}+j l_{21I} & l_{22} 
\end{smallmatrix} \right)$$
and
$$  \mathbf{Q}=\left( \begin{smallmatrix} 
  e^{j a_{1}}\cos(\theta) & e^{j a_{2}}\sin(\theta)\\
  -e^{j(a_{3}-a_{2})}\sin(\theta) & e^{j(a_{3}-a_{1})}\cos(\theta) 
\end{smallmatrix} \right)$$ with: $l_{21R},l_{21I} \in \mathbb{R}$, $ l_{11}, l_{22} \ge 0 $, $0 \leqslant \theta \leqslant \pi/2$ and $0 \leqslant a_{1},a_{2},a_{3} \leqslant 2\pi$.
The associated Jacobian is $l_{11}^{3}l_{22}\sin(\theta)\cos(\theta)$. This, along with the PDF of $\mathbf{H}$
\begin{align}
p(\mathbf{H}) &= \frac{K}{\pi^4}\prod_{1\le i,j \le 2}  e^{ - |h_{ij}|^2 / \phi_{ij}  }, \nonumber
\end{align}
gives an expression for the joint PDF of $(\mathbf{L}, \mathbf{Q})$.
The marginal PDF of $\mathbf{L}$ then follows by integrating over the parameters $0 \leqslant \theta \leqslant \pi/2$, $0 \leqslant a_{1},a_{2},a_{3} \leqslant 2\pi$, which leads to (\ref{eq:integral_L}) (top of the next page), where
\begin{align}
{\cal I}_\theta =  \int_{0}^{2\pi}\int_{0}^{2\pi}\int_{0}^{2\pi} e^{p_\theta(a_{2},a_{3}) \cos{a_{1}} + q_\theta(a_{2},a_{3}) \sin{a_{1}} } da_1 da_2 da_3  \nonumber
\end{align}
with
\begin{align}
p_\theta(a_{2},a_{3}) &= l_{22}\sin(2\theta)\left(\frac{1}{\phi_{21}}-\frac{1}{\phi_{22}}\right) \nonumber \\
&  \times (l_{21R}\cos{(a_{3}-a_{2})}+l_{21I}\sin{(a_{3}-a_{2})}) \nonumber  \\
q_\theta(a_{2},a_{3}) &= l_{22}\sin(2\theta)\left(\frac{1}{\phi_{21}}-\frac{1}{\phi_{22}}\right) \nonumber \\
&  \times (l_{21R}\sin{(a_{3}-a_{2})}-l_{21I}\cos{(a_{3}-a_{2})}) \; . \nonumber 
\end{align}
\begin{figure*}[htb]
\normalsize
\setcounter{equation}{12}
\begin{align}
p(\mathbf{L})&=\frac{l_{11}^{3}l_{22}}{2\pi^{4}} K \int_0^{\pi/2}   e^{ -l_{11}^{2} \left(\frac{\cos^{2}(\theta)}{\phi_{11}} + \frac{\sin^{2}(\theta)}{\phi_{12}}\right) -|l_{21}|^{2} \left(\frac{\cos^{2}(\theta)}{\phi_{21}} + \frac{\sin^{2}(\theta)}{\phi_{22}}\right) - l_{22}^{2}\left( \frac{\sin^{2}(\theta)}{\phi_{21}}+\frac{\cos^{2}(\theta)}{\phi_{22}}\right) }
 \sin(2\theta)  \,{{\cal I}_\theta} \, d\theta \label{eq:integral_L}
\end{align}
\vspace*{4pt}
\hrulefill
\end{figure*}
We now show that the integral ${\cal I}_\theta$ can be computed explicitly. The result \cite[Eq. 3.937.2]{gradshteyn1965table} allows us to integrate over $a_1$ giving
\begin{align}
{\cal I}_{\theta} 
=&\int_{[0,2\pi]^2} 2\pi I_{0}(\sqrt{p_\theta(a_{2},a_{3})^{2}+q_\theta(a_{2},a_{3})^{2}})  d a_{2} d a_{3}, \nonumber  
\end{align}
where $I_{0}\left( \cdot \right)$ is a modified Bessel function of the first kind.  Using the expansion, $I_{0}\left( x \right)=\sum_{k=0}^{\infty}{\frac{\left( x^2 / 4 \right)^{k}}{ (k!)^2 }}$ and noting that $p_\theta(a_{2},a_{3})^{2}+q_\theta(a_{2},a_{3})^{2}=\left(|l_{21}|l_{22}\sin(2\theta)\left(\frac{1}{\phi_{21}}-\frac{1}{\phi_{22}}\right)\right)^{2}$ does not depend on $a_{2},a_{3}$, the remaining integrals are easy to evaluate, giving
\begin{align*}
{\cal I}_{\theta}=(2\pi)^{3}\sum_{k=0}^{\infty}  \frac{\left( |l_{21}|l_{22}\sin{(2\theta)\begin{vmatrix}\frac{1}{\phi_{21}}-\frac{1}{\phi_{22}}\end{vmatrix} / 2 }\right)^{2k}}{(k!)^{2}}.
\end{align*}
Substituting into (\ref{eq:integral_L}), rearranging terms and applying double-angle formulae, the PDF of $\mathbf{L}$ takes the form
\begin{align}
p(\mathbf{L})=\frac{4l_{11}^{3}l_{22}}{\pi} K  \e^{- \alpha}\sum_{k=0}^{\infty} \frac{\left[|l_{21}|l_{22}|\epsilon_{2}|\right]^{2k}}{2^{2k}(k!)^{2}} {\cal J}_k    \label{eq:integral_L2}
\end{align}
where
\begin{align}
{\cal J}_k = \int_{0}^{\pi/2}\e^{- \beta \cos(2\theta)} \left(\sin(2\theta)\right)^{2k+1}d\theta \nonumber
\end{align}
and
\begin{align}
& \alpha = \frac{ l_{11}^{2}s_{1} +(|l_{21}|^{2}+l_{22}^{2})s_{2}}{2}, \; \; \beta = \frac{l_{11}^{2}\epsilon_{1}+(|l_{21}|^{2}-l_{22}^{2})\epsilon_{2}}{2} \nonumber
\end{align}
with $s_i$ and $\epsilon_i$ defined in the theorem statement.
Now, using \cite[Eq.\ 3.915.4]{gradshteyn1965table}, the remaining integral evaluates to:
\begin{align*}
{\cal J}_{k} = \frac{\sqrt{\pi}}{2}k!\sum_{m=0}^{\infty}{\frac{\beta^{2m}}{2^{2m}m!\Gamma\left(k+m+\frac{3}{2}\right)}}.
\end{align*}
Substituting into (\ref{eq:integral_L2}) and rearranging the terms, we get
\begin{align}
p(\mathbf{L}) =\frac{2l_{11}^{3}l_{22}}{\sqrt{\pi}}e^{- \alpha} K \sum_{k=0}^{\infty}\sum_{m=0}^{\infty}\frac{\left( \frac{|l_{21}|^{2}l_{22}^{2}\epsilon_{2}^{2}}{4}\right)^{k}\left(\frac{\beta^{2}}{4}\right)^{m}}{k!m!\Gamma\left(k+m+\frac{3}{2}\right)}  \;  \nonumber
\end{align}
which, upon substituting for $\alpha$ and $\beta$ and noting that $\sum_{k=0}^{\infty}\sum_{m=0}^{\infty} \frac{a^{k}b^{m}}{k!m!\Gamma{\left(k+m+\frac{3}{2}\right)}} = \frac{\sinh{\left(2\sqrt{a+b}\right)}}{\sqrt{\pi}\sqrt{a+b}}$ yields the desired PDF of $\mathbf{L}$,
\begin{align}
p(\mathbf{L})&=\frac{8l_{11}^3 l_{22}}{\pi} K e^{-\frac{1}{2}\left(l_{11}^2 s_{1}+\left(|l_{21}|^2 +l_{22}^2\right)s_2 \right)} \nonumber \\
& \hspace{-10mm} \times \frac{\sinh\left(\frac{1}{2}\sqrt{(l_{11}^2\epsilon_1+(|l_{21}|^2+l_{22}^2)\epsilon_2)^2-4 l_{11}^2 l_{22}^2\epsilon_{1}\epsilon_{2}}\right)}{\sqrt{(l_{11}^2\epsilon_1+(|l_{21}|^2+l_{22}^2)\epsilon_2)^2-4 l_{11}^2 l_{22}^2\epsilon_1\epsilon_2}}. \label{eq:dis_L}
\end{align}
Now consider the Cholesky decomposition $\mathbf{W}=\mathbf{H}\mathbf{H}^{\dagger}= \mathbf{L}\mathbf{L}^{\dagger}$.  
The Jacobian of this transformation is $4l_{11}^{3}l_{22}$. With this, along with (\ref{eq:dis_L}), we perform the change of variables $\mathbf{L} \to \mathbf{W}$ to arrive at the result stated in (\ref{eq:pw1}).

\subsection{Proof of Theorem \ref{th:lambda} } \label{Ap:Th2Proof}

We consider the eigendecomposition $\mathbf{W}=\mathbf{U} \mathbf{\Lambda} \mathbf{U}^{\dagger}$, with $\mathbf{\Lambda}=\left( \begin{smallmatrix} 
  \lambda_{1} & 0\\
  0 & \lambda_{2} 
\end{smallmatrix} \right)$, $\lambda_1 \ge \lambda_2>0$, and $\mathbf{U}$ unitary.  
For $2 \times 2$, $\mathbf{U}$ may be parametrised in terms of 2 free (angular) parameters as follows (see e.g., \cite{fraidenraich2007mimo,fraidenraich2008mimo}): $\mathbf{U}=\left( \begin{smallmatrix} 
  \cos{\kappa} & -e^{j\psi}\sin{\kappa}\\
  e^{-j\psi}\sin{\kappa} & \cos{\kappa} 
\end{smallmatrix} \right)$, $0 \leqslant \kappa \leqslant \pi/2$, $0 \leqslant \psi \leqslant 2\pi$. 
The Jacobian of this transformation is $\frac{1}{2}(\lambda_{1}-\lambda_{2})^{2}\sin(2\kappa)$. Identifying the entries of $\mathbf{W}$ in (\ref{eq:pw1}) in terms of the four real parameters $\kappa, \psi, \lambda_1$ and $\lambda_2$, we obtain the following  for the ordered eigenvalue PDF of $\mathbf{W}$:
\begin{align}
p(\lambda_{1},\lambda_{2})&= \frac{K}{\pi} \, (\lambda_{1}-\lambda_{2})^{2}  \int_0^{2 \pi} \int_{0}^{\frac{\pi}{2}}e^{-\frac12 \nu(\lambda_{1},\lambda_{2},\kappa)} \, \nonumber \\
& \times \frac{\sinh\left(\frac12 \sqrt{\eta(\lambda_{1},\lambda_{2},\kappa)}\right)}{\sqrt{\eta(\lambda_{1},\lambda_{2},\kappa)}}\sin(2\kappa)d\kappa d\psi  \, , \nonumber
\end{align}
where $\nu(\lambda_{1},\lambda_{2},\kappa)$ and $\eta(\lambda_{1},\lambda_{2},\kappa)$ are given in the theorem statement.
After integration over $\psi$, we obtain (\ref{eq:plambda_integral}).

\subsection{Proof of Corollary \ref{cor:epsilon}}
\label{app:sec}
The PDF of $\mathbf{W}$ (\ref{eq:w_part}) follows directly from (\ref{eq:pw1}). The joint eigenvalue PDF (\ref{eq:lambda_part_case}) is obtained from (\ref{eq:w_part}) upon applying the same steps as in Appendix \ref{Ap:Th2Proof}.
The expression (\ref{eq:lambda_part_case}) is particularly convenient, allowing us to use the Cauchy-Binet formula \cite{andreief1883note} to compute the marginal CDFs of the extreme eigenvalues of $\mathbf{W}$. For example, for $\phi_1 \neq \phi_2$, the CDF of the minimum eigenvalue admits
\begin{align}
F_{\lambda_\mathrm{min}}(x) &  = \mathbb{P}(\lambda_\mathrm{min}\le x)
 = 1- \int_x^\infty \int_x^{\lambda_1} p(\lambda_1,\lambda_2) d\lambda_2 d\lambda_1 \nonumber\\
&  \hspace{-5mm}= 1- \frac{ 1}{|\epsilon_1|\phi_3^2} \det \left( \int_x^\infty \lambda^{i-1}g(\lambda)^{j-1} e^{-\lambda/\phi_3} d\lambda \right)_{i,j=1,2}  \nonumber\\
& \hspace{-5mm} =1-\frac{e^{-x/\phi_3}}{\phi_2 - \phi_1}\left( \phi_2e^{-x/\phi_2} - \phi_1e^{-x/\phi_1} \right. \nonumber \\
&  \hspace{-4mm} \Bigl. + x \left( \Ei(-x/\phi_2)- \Ei(-x/\phi_1) \right)\Bigr),  \nonumber 
\end{align}
where the Cauchy-Binet formula was applied to obtain the third equality, while the last equality was obtained using \cite[Eq.\ 5.231.2]{gradshteyn1965table} and integration by parts. For the maximum eigenvalue, $F_{\lambda_\mathrm{max}}(x)  = \int_0^x  \int_0^{\lambda_1} p(\lambda_1,\lambda_2) d\lambda_2 d\lambda_1$, which is computed using analogous steps. 

\bibliographystyle{IEEEbib}
\bibliography{BibTeX1}

\end{document}